\documentclass[conference]{IEEEtran}
\usepackage{amssymb}
\usepackage{amsmath}
\usepackage{amsfonts}
\usepackage{graphicx}
\usepackage{algorithm}
\usepackage{algorithmic}
\usepackage{cite}
\usepackage{epstopdf}
\usepackage{stfloats}

\newtheorem{theorem}{\textbf{Theorem}}

\newtheorem{lemma}{\textbf{Lemma}}

\newtheorem{proposition}{\textbf{Proposition}}
\newtheorem{remark}{\textbf{Remark}}

\begin{document}

\title{Wireless Information and Power Transfer in Two-Way Amplify-and-Forward Relaying Channels}
\author{Zhiyong Chen, Biao Wang, Bin Xia and Hui Liu \\
Department of Electronic Engineering, Shanghai Jiao Tong University, Shanghai, P. R. China\\
Email: {\{zhiyongchen, bxia, huiliu\}@sjtu.edu.cn}}
\maketitle
\begin{abstract}
The various wireless networks have made the ambient radio frequency signals around the world. Wireless information and power transfer enables the devices to recycle energy from these ambient radio frequency signals and process information simultaneously. In this paper, we develop a wireless information and power transfer protocol in two-way amplify-and-forward relaying channels, where two sources exchange information via an energy harvesting relay node. The relay node collects energy from the received signals and uses it to provide the transmission power to forward the received signals. We analytically derive the exact expressions of the outage probability, the ergodic capacity and the finite-SNR diversity-multiplexing trade-off (DMT). Furthermore, the tight closed-form upper and lower bounds of the outage probability and the ergodic capacity are then developed. Moreover, the impact of the power splitting ratio is also evaluated and analyzed. Finally, we show that compared to the non-cooperative relaying scheme, the proposed protocol is a green solution to offer higher transmission rate and more reliable communication without consuming additional resource.
\end{abstract}
\section{Introduction}
 As various wireless networks developed, most devices are surrounded by ambient radio frequency (RF) signals anytime and anywhere, e.g., cellular signals or Wi-Fi signals. Every ambient RF signal carries not only information but also energy. It has been shown that one device can wirelessly recycle these energy from the ambient RF signals \cite{ARF_Sigcom}. Recently, the wireless information and power transfer technology enables one device to collect energy and process the information from the ambient RF signals simultaneously \cite{First_paper_energy,WIPT_TWC}. This offers an exciting new way to meet the requirement of the green communications.

 The basic idea of wireless information and power transfer was first proposed in \cite{First_paper_energy} and a general receiver architecture was then developed in \cite{WIPT_TWC}. Following these two pioneering works, the concept was extended to multiple-input multiple-output (MIMO) systems in \cite{WIPT_TWC_MIMO, WIPT_IWCT}, cooperative networks in \cite{WIPT_TWC_Relay}, and orthogonal frequency division multiplexing (OFDM) systems in \cite{WIPT_OFDM}, etc. On the other hand, the rate-energy tradeoff was analyzed in \cite{First_paper_energy,WIPT_TWC} and  the outage probability and throughput were analyzed for the one-way relay channels in \cite{WIPT_TWC_Relay}. An energy-efficient power allocation scheme for cooperative networks was developed in \cite{WIPT_Power}.

 In this paper, we design the wireless information and power transfer protocol for two-way relay channels by using the amplify-and-forward (AF) scheme, where two source nodes exchange information through an energy constrained relay node. The wireless information and power transfer enables the relay node to deliver both sources' signals without any itself energy. We find that the energy constrained relay node cannot purely enlarge the signal quality because of the fact that it does not consume any extra energy. Subsequently, the exact expressions of the outage probability, the ergodic capacity and the finite-SNR diversity-multiplexing trade-off (DMT) for the proposed protocol are derived. Further, we develop the tight closed-form upper and lower bounds of the outage probability and the ergodic capacity of the system. Moreover, the impact of the power splitting ratio depicted the trade-off between the harvesting energy and the forward signals' power on the ergodic capacity and finite-SNR DMT is also evaluated and analyzed in this paper. Finally, we show that the use of the relay node can improve the ergodic capacity and achieve higher outage performance. Surprisingly, this improvement is not build on the additional resource consumption, neither the energy, time nor radio spectrum resource.
\section{System model and Protocol description}
\subsection{System Model}
We consider a half-duplex two-way relay channels where sources $S_{1}$ and $S_{2}$ exchange information through an energy harvesting relay node. By assumption, the relay node only forwards the data and is not a source or destination. Theoretically, the relay operations are usually carried out in two transmission stages, namely, the multiple-access (MA) stage (e.g., sources-to-relay) and the broadcasting (BC) stage (e.g., relay-to-destinations). In the MA stage, $S_{1}$ and $S_{2}$ transmit their messages to the relay node simultaneously. The resulting signals are then broadcast to $S_{1}$ and $S_{2}$ in the BC stage. The signal processing of the relay node is based on the AF scheme and the duration of both stages can thus be normalized to $1/2$.

The system operates in joint information and power transfer model. Specifically, both source nodes have a fixed power supply, i.e., the transmit power of $S_{1}$ and $S_{2}$ is $P_{1}$ and $P_{2}$, respectively. For an energy constrained relay node, however, there is no fixed power supply and it thus needs to scavenge energy from the received signal in the MA stage. Furthermore, the relay node simultaneously processes the information using the harvesting energy.
\subsection{Energy Harvesting Relaying Protocol}
During the MA stage, the received signal at the relay node is given by
\begin{equation}
y_{r}=\sqrt{P_{1}}h_{1}x_{1}+\sqrt{P_{2}}h_{2}x_{2}+\widetilde{n}_{a},
\end{equation}
where $x_{i}$ is the unit-power transmitted information, $\widetilde{n}_{a}$ denotes the narrow-band Gaussian noise introduced by the receiving antenna, $h_{i}$ characterizes the channel gain between $S_{i}$ and the relay node. All channels are modeled as the quasi-static Rayleigh fading channel in this paper and thus we have $h_{i}\thicksim \mathcal{CN}(0,\Omega_{i})$, $i=1,2$.

As described in \cite{WIPT_TWC, WIPT_TWC_Relay}, the power splitting model is used at the energy constrained relay node. The received signals' power can be split into two parts by a power splitter, one for energy harvesting and the other one for information processing. The signal for energy harvesting can be expressed as
\begin{equation}
\sqrt{\lambda}y_{r}=\sqrt{\lambda P_{1}}h_{1}x_{1}+\sqrt{\lambda P_{2}}h_{2}x_{2}+\sqrt{\lambda}\widetilde{n}_{a},
\end{equation}
where $0<\lambda<1$ is the portion signal power split to scavenge energy. Here, we can calculate the total harvested energy during the MA stage as following \cite{WIPT_TWC}
\begin{equation}
Q=\eta\lambda(P_{1}|h_{1}|^{2}+P_{2}|h_{2}|^{2})\frac{1}{2}.
\end{equation}
Here, $0<\eta\leqslant1$ denotes the energy conversion efficiency. We then have the transmitted power of the relay node as following:
\begin{equation}
P_{r}=Q/(1/2)=\eta\lambda(P_{1}|h_{1}|^{2}+P_{2}|h_{2}|^{2}).
\end{equation}

Meanwhile, the remaining received signal power is sent to do the information processing. The broadcasting signal by the relay node is then given by
\begin{align}
x_{r}{\setlength\arraycolsep{0.5pt}=}\beta\sqrt{P_{r}}(\sqrt{1-\lambda}y_{r}{\setlength\arraycolsep{0.5pt}+}n_{b}){\setlength\arraycolsep{0.5pt}\approx } \sqrt{\frac{\eta\lambda}{1-\lambda}}(\sqrt{1-\lambda}y_{r}+n_{b}),
\end{align}
where $n_{b}\thicksim \mathcal{CN}(0,\sigma^{2}_{b})$ denotes the additive white Gaussian noise (AWGN) introduced by the signal conversion from passband to baseband. The power constraint factor $\beta$ of the relay node is given by
\begin{align}
\beta&=\frac{1}{\sqrt{(1-\lambda)P_{1}|h_{1}|^{2}+(1-\lambda)P_{2}|h_{2}|^{2}+(1-\lambda)\sigma_{a}^{2}+\sigma_{b}^{2}}}\nonumber\\
&\approx\frac{1}{\sqrt{(1-\lambda)P_{1}|h_{1}|^{2}+(1-\lambda)P_{2}|h_{2}|^{2}}}.
\end{align}
Here, the passband noise $\widetilde{n}_{a}$ is changed to the baseband AWGN $n_{a}\thicksim \mathcal{CN}(0,\sigma^{2}_{a})$.

During the BC stage, the received signals at $S_{i}$ is given by
\begin{align}
&y_{i}=h_{i}x_{r}+n_{i}\\
&{\setlength\arraycolsep{0.3pt}=}\sqrt{\eta\lambda P_{i}}h_{i}^{2}x_{i}{\setlength\arraycolsep{0.3pt}+}\sqrt{\eta\lambda P_{j}}h_{i}h_{j}x_{j}{\setlength\arraycolsep{0.3pt}+}\sqrt{\eta\lambda}h_{i}n_{a}{\setlength\arraycolsep{0.3pt}+}\sqrt{\frac{\eta\lambda}{1{\setlength\arraycolsep{0.3pt}-}\lambda}}h_{i}n_{b}{\setlength\arraycolsep{0.3pt}+}n_{i},\nonumber
\end{align}
for $i, j\in\{1,2\}$ and $i\neq j$. We assume the channels are reciprocal and $n_{i}\thicksim \mathcal{CN}(0,\sigma^2_{i})$ is noise at $S_{i}$. Without loss of generality, we assume $\sigma^{2}_{1}=\sigma^{2}_{2}=\sigma^{2}_{a}+\sigma^{2}_{b}=\sigma^{2}$.

Since $x_{i}$ is known perfectly by $S_{i}$, $S_{i}$ can cancel the self interference from $y_{i}$. Therefore, we can compute the signal-to-noise ratio (SNR) as following
\begin{align}
\gamma_{1}=\frac{P_{2}|h_{2}|^{2}}{\sigma^{2}}\frac{1}{1{\setlength\arraycolsep{0.3pt}+}\frac{\epsilon\lambda}{1{\setlength\arraycolsep{0.3pt}-}\lambda}{\setlength\arraycolsep{0.3pt}+}\frac{1}{\eta\lambda|h_{1}|^{2}}}\label{SNR1}
\end{align}
for $S_{1}$ and
\begin{align}
\gamma_{2}=\frac{P_{1}|h_{1}|^{2}}{\sigma^{2}}\frac{1}{1+\frac{\epsilon\lambda}{1-\lambda}+\frac{1}{\eta\lambda|h_{2}|^{2}}}\label{SNR2}
\end{align}
for $S_{2}$. Here, we use $\epsilon=\sigma^{2}_{b}/(\sigma^{2}_{a}+\sigma^{2}_{b})$.

Accordingly, the data rate at $S_{i}$ is given by
\begin{align}
R_{i}=\frac{1}{2}\log_{2}(1+\gamma_{i}).
\end{align}
\begin{remark}
From (\ref{SNR1}) and (\ref{SNR2}), it is worth mentioning that $\gamma_{1}\leq P_{2}|h_{2}|^{2}/\sigma^{2}$ and $\gamma_{2}\leq P_{1}|h_{1}|^{2}/\sigma^{2}$ always hold. We notice that $P_{i}|h_{i}|^{2}/\sigma^{2}$ is the SNR of the channel between $S_{i}$ and the relay node in the MA stage. It is implied that the energy constrained relay node cannot purely enlarge the signal quality without consuming additional energy. However, as shown in Section IV, considering the impact of the path loss between two nodes, such protocol can also outperform the non-cooperative relaying scheme, where both sources communicate directly without cooperating with one relay node.
\end{remark}
\section{Information-Theoretic Metrics}
In this section, we investigate the performance of the proposed energy harvesting relaying protocol. Our information-theoretic metrics of interest are the outage probability, the ergodic capacity and the finite-SNR diversity-multiplexing trade-off.
\subsection{Outage Probability}
For two-way AF relaying channels, the overall system outage probability is defined as
\begin{align}
P_{out}&=\Pr\left(R_{1}<T_{1},~~or~~ R_{2}<T_{2}\right)\label{outage_definition}\\
&=\Pr(\gamma_{1}<\tau_{1})+\Pr(\gamma_{2}<\tau_{2})-\Pr\left(\gamma_{1}<\tau_{1}, \gamma_{2}<\tau_{2}\right)\nonumber,
\end{align}
where $T_{i}$ denotes the target rate of $S_{i}$ and we use $\tau_{i}=2^{2T_{i}}-1$, for $i=1,2$.

We can rewrite the output SNR as following
\begin{align}
\gamma_{1}=\frac{P_{2}}{\sigma^2}\frac{|h_{1}|^{2}|h_{2}|^{2}}{b|h_{1}|^{2}+c},~~~\gamma_{2}=\frac{P_{1}}{\sigma^2}\frac{|h_{1}|^{2}|h_{2}|^{2}}{b|h_{2}|^{2}+c}.
\end{align}
where $b=1+\epsilon\lambda/(1-\lambda)$ and $c=1/\eta\lambda$. In order to derive the outage probability, we first establish the following lemma.
\begin{lemma}
 Let $X$ and $Y$ be $|h_{1}|^2$ and $|h_{2}|^2$, respectively. Let us define the following variable
\begin{align}
Z=\frac{aXY}{bX+c}, a\geq0, b\geq0,c\geq0.
\end{align}
where $a$, $b$ and $c$ are independent parameters with $X$ and $Y$. We thus have the cumulative density function (cdf) of $Z$ as
\begin{align}
F_{Z}(z)=1-\frac{\exp(-\frac{zb}{a\Omega_{2}})}{\Omega_{1}}\sqrt{\frac{4zc\Omega_{1}}{a\Omega_{2}}}\mathcal{K}_{1}\left(\sqrt{\frac{4zc}{a\Omega_{1}\Omega_{2}}}\right),
\end{align}
where $\mathcal{K}_{n}(\cdot)$ denotes the modified Bessel function of the second kind with order $n$ defined in \cite{Math_book}.
\end{lemma}
\begin{IEEEproof}
By definition, it is easy to prove this lemma.
\end{IEEEproof}

 Using the cdf in Lemma 1, the following theorem describes an exact expression of the outage probability.
\begin{theorem}
 The outage probability of two-way AF relaying channels with energy harvesting can be expressed as
 \begin{align}
P_{out}&\approx1+\exp{\left(-\frac{X_{0}}{\Omega_{1}}-\frac{Y_{0}}{\Omega_{2}}\right)}\label{outage_exact}\\
&-\sum\limits_{ i,j \in \{ 1,2\}  \hfill \atop
i \ne j }\left\{\exp({\setlength\arraycolsep{0.5pt}-}\frac{\sigma^2\tau_{i}b}{P_{j}\Omega_{j}})\sqrt{\frac{4\sigma^2\tau_{i}c}{P_{j}\Omega_{1}\Omega_{2}}}\mathcal{K}_{1}\left(\sqrt{\frac{4\sigma^2\tau_{i}c}{P_{j}\Omega_{1}\Omega_{2}}}\right)\right.\nonumber\\
&\left.-\sum_{n=0}^{2}\frac{\mathcal{H}^{(n)}_{i,j}(\nu_{i})}{n!(n+1)}\Bigg\{(V_{i}-\nu_{i})^{n+1}-(-\nu_{i})^{n+1}\Bigg\}\right\}\nonumber,
\end{align}
where $\mathcal{H}_{i,j}(z)=\exp\left(\frac{-\sigma^2\tau_{j}c}{P_{i}\Omega_{i}}\frac{1}{z}-\frac{z}{\Omega_{j}}\right)$, $\mathcal{H}^{(0)}_{i,j}(\nu_{i})=\mathcal{H}_{i,j}(\nu_{i})$, $\mathcal{H}^{(1)}_{i,j}(\nu_{i})=\mathcal{H}_{i,j}(\nu_{i})(\frac{\sigma^2\tau_{j}c}{P_{i}\Omega_{i}\nu_{i}^2}-\frac{1}{\Omega_{j}})$ and $\mathcal{H}^{(2)}_{i,j}(\nu_{i})=\mathcal{H}_{i,j}(\nu_{i})\left\{\left(\frac{\sigma^2\tau_{j}c}{P_{i}\Omega_{i}\nu_{i}^2}-\frac{1}{\Omega_{j}}\right)^2-\frac{2\sigma^2\tau_{j}c}{P_{i}\Omega_{i}\nu_{i}^3}\right\}$. Here, $\nu_{i}=V_{i}/2$ with $V_{1}=Y_{0}$ and $V_{2}=X_{0}$. $X_{0}$ and $Y_{0}$ are given by $(\varphi_{1}+\sqrt{\varphi^2_{1}+(4\sigma^2\tau_{1}^2\tau_{2}b^2c)/P_{1}})/2\tau_{1}b$ and $(\varphi_{2}+\sqrt{\varphi^2_{2}+(4\sigma^2\tau_{2}^2\tau_{1}b^2c)/P_{2}})/2\tau_{2}b$ with $\varphi_{1}=\sigma^2\tau_{1}\tau_{2}b^{2}/P_{1}+P_{2}\tau_{2}c/P_{1}-\tau_{1}c$ and $\varphi_{2}=\sigma^2\tau_{1}\tau_{2}b^{2}/P_{2}+P_{2}\tau_{2}c/P_{2}-\tau_{2}c$, respectively.
\end{theorem}
\begin{proof}
See Appendix A.
\end{proof}

As shown in the preceding theorem, the integral item can be calculated through numerical computation and then we have $P_{out}$. In order to reduce the computation complexity, lower and upper bounds of $P_{out}$ are derived as following proposition.
\begin{proposition}
The outage probability can be lower bounded by
\begin{align}
P_{out}&\geq1+\exp{\left(-\frac{X_{0}}{\Omega_{1}}-\frac{Y_{0}}{\Omega_{2}}\right)}-\exp\left(-\frac{\sigma^2\tau_{2}b}{P_{1}\Omega_{1}}-\frac{Y_{0}}{\Omega_{2}}\right)\nonumber\\
&-\exp\left(-\frac{\sigma^2\tau_{1}b}{P_{2}\Omega_{2}}-\frac{X_{0}}{\Omega_{1}}\right),\label{lowerbound}
\end{align}
and can also be upper bounded as follows
 \begin{align}
P_{out}&\leq1+\exp{\left(-\frac{X_{0}}{\Omega_{1}}-\frac{Y_{0}}{\Omega_{2}}\right)}-\exp\left(-\frac{\sigma^2\tau_{2}b}{P_{1}\Omega_{1}}-\frac{Y_{0}}{\Omega_{2}}\right)\nonumber\\
&{\setlength\arraycolsep{0.5pt}\times}\sqrt{\frac{4\sigma^2\tau_{2}c}{P_{1}\Omega_{1}\Omega_{2}}}\mathcal{K}_{1}\left(\sqrt{\frac{4\sigma^2\tau_{2}c}{P_{1}\Omega_{1}\Omega_{2}}}\right){\setlength\arraycolsep{0.5pt}-}\exp\left({\setlength\arraycolsep{0.5pt}-}\frac{\sigma^2\tau_{1}b}{P_{2}\Omega_{2}}{\setlength\arraycolsep{0.5pt}-}\frac{X_{0}}{\Omega_{1}}\right)\nonumber\\
&\times\sqrt{\frac{4\sigma^2\tau_{1}c}{P_{2}\Omega_{1}\Omega_{2}}}\mathcal{K}_{1}\left(\sqrt{\frac{4\sigma^2\tau_{1}c}{P_{2}\Omega_{1}\Omega_{2}}}\right).\label{upperbound}
\end{align}
\end{proposition}
\begin{IEEEproof}
See Appendix B.
\end{IEEEproof}

It is easily seen that $X_{0}\rightarrow 0$ and $Y_{0}\rightarrow 0$ for high SNR. Meanwhile, the modified bessel function of the second kind is bounded as \cite{PNC_TWC_MIMO}
\begin{align}
\exp\left(-x\right)\leq x\mathcal{K}_{1}(x)\leq1.\label{bound_K1}
\end{align}
Hence we have $x\mathcal{K}_{1}(x)\rightarrow 1$ when $x\rightarrow0$ based on Squeeze Theorem. As a result, for high SNR, the exact expression of outage probability $P_{out}$ in (\ref{outage_exact}), the lower bound in (\ref{lowerbound}) and the upper bound in (\ref{upperbound}) are all approximated as
\begin{align}
P_{out}\approx2{\setlength\arraycolsep{0.5pt}-}\exp\left({\setlength\arraycolsep{0.5pt}-}\frac{\sigma^2\tau_{2}b}{P_{1}\Omega_{1}}\right){\setlength\arraycolsep{0.5pt}-}\exp\left({\setlength\arraycolsep{0.5pt}-}\frac{\sigma^2\tau_{1}b}{P_{2}\Omega_{2}}\right).
\end{align}
\subsection{Ergodic Capacity}
Now let us derive the ergodic capacity for two-way AF relaying channels with energy harvesting. The total ergodic capacity can be given by
\begin{align}
C_{e}&=\mathbb{E}_{h_{1},h_{2}}\{\frac{1}{2}\log_{2}(1+\gamma_{1})\}+\mathbb{E}_{h_{1},h_{2}}\{\frac{1}{2}\log_{2}(1+\gamma_{2})\}\nonumber\\
&\overset{(i)}{{\setlength\arraycolsep{0.5pt}=}}\frac{1}{2\ln2}\int_{0}^{\infty}\frac{1{\setlength\arraycolsep{0.5pt}-}F_{1}(z)}{1+z}dz{\setlength\arraycolsep{0.5pt}+}\frac{1}{2\ln2}\int_{0}^{\infty}\frac{1{\setlength\arraycolsep{0.5pt}-}F_{2}(z)}{1+z}dz,
\end{align}
where Step (i) is based on the integration by parts and $F_{i}(\cdot)$ is the cdf of $\gamma_{i}$. We then obtain the following theorem.
\begin{theorem}
The ergodic capacity of two-way AF relaying channels with energy harvesting is
 \begin{align}
&C_{e}{\setlength\arraycolsep{0.5pt}=}\frac{1}{2\ln2}\sum_{i=1}^{2}\Psi\left(1,1;\frac{\sigma^2b}{P_{i}\Omega_{i}}\right){\setlength\arraycolsep{0.5pt}+}\frac{1}{2\ln2}\sum_{i=1}^{2}\left\{\sum_{l=0}^{\infty}\frac{(\sigma^2c)^{l+1}}{(P_{i}\Omega_{1}\Omega_{2})^{l+1}l!}\right.\nonumber\\
&{\setlength\arraycolsep{0.3pt}\times}\left[\left(\ln\frac{\sigma^2c}{P_{i}\Omega_{1}\Omega_{2}}{\setlength\arraycolsep{0.5pt}+}2\mathbf{C}{\setlength\arraycolsep{0.5pt}-}\sum_{k=1}^{l}{\frac{1}{k}}{\setlength\arraycolsep{0.5pt}-}\sum_{k=1}^{l+1}{\frac{1}{k}}\right)\Psi\left(l{\setlength\arraycolsep{0.3pt+}}2,l{\setlength\arraycolsep{0.3pt+}}2;\frac{\sigma^2b}{P_{i}\Omega_{i}}\right)\right.\nonumber\\
&\left.{\setlength\arraycolsep{0.3pt}+}\frac{1}{(l+1)!}\mathcal{J}_{i,l}\right]+\frac{\sigma^2c}{P_{i}\Omega_{1}\Omega_{2}}\left(\ln\frac{\sigma^2c}{P_{i}\Omega_{1}\Omega_{2}}{\setlength\arraycolsep{0.5pt}+}2\mathbf{C}{\setlength\arraycolsep{0.5pt}-}1\right)\nonumber\\
&\left.\times\Psi\left(2,2;\frac{\sigma^2b}{P_{i}\Omega_{i}}\right)+\frac{\sigma^2c}{P_{i}\Omega_{1}\Omega_{2}}\mathcal{J}_{i,0}\right\},\label{bound_EC}
\end{align}
where $\Psi(\alpha,\beta;z)$ is the confluent hypergeometric function \cite{Math_book} and $\mathbf{C}\approx0.5772$ is Euler's constant. Here we have 
\begin{align}
\mathcal{J}_{i,l}=\int_{0}^{\infty}\frac{\exp\left(\frac{{\setlength\arraycolsep{0.3pt}-}\sigma^2bz}{P_{i}\Omega_{i}}\right)z^{l{\setlength\arraycolsep{0.3pt}{\setlength\arraycolsep{0.3pt}+}}1}\ln{z}}{1+z}dz.
\end{align}

\end{theorem}
\begin{proof}
See Appendix C.
\end{proof}
\newcounter{TempEqCnt}
\setcounter{TempEqCnt}{\value{equation}}
\setcounter{equation}{27}
\begin{figure*}[hb]
\hrulefill
\begin{align}
d(r,\gamma)\thickapprox\frac{\gamma\left\{A\left(\frac{1}{\Omega_{1}}+\frac{1}{\Omega_{2}}\right)\exp\left(-(\frac{1}{\Omega_{1}}+\frac{1}{\Omega_{2}})X_{0}\right)-\sum\limits_{ i,j \in \{ 1,2\}  \hfill \atop
i \ne j }\left(\frac{Bb}{\Omega_{i}}+\frac{A}{\Omega_{j}}\right)\exp\left(-\frac{b(1+\gamma)^{r}-b}{\gamma\Omega_{i}}-\frac{X_{0}}{\Omega_{j}}\right)\right\}}{1+\exp\left(-(\frac{1}{\Omega_{1}}+\frac{1}{\Omega_{2}})X_{0}\right)-\sum\limits_{ i,j \in \{ 1,2\}  \hfill \atop
i \ne j }\exp\left(-\frac{b(1+\gamma)^{r}-b}{\gamma\Omega_{i}}-\frac{X_{0}}{\Omega_{j}}\right)}.\label{DMT}
\end{align}
\end{figure*}
\setcounter{equation}{\value{TempEqCnt}}

Since $\mathcal{J}_{i,l}$ cannot be obtained in a closed form, we develop the following proposition to fast and efficiently evaluate the ergodic capacity.
\begin{proposition}
The ergodic capacity of two-way AF relaying channels with energy harvesting is bounded as following
\begin{align}
&\int_{0}^{\infty}\frac{\exp\left(\frac{-\sigma^2bz}{P_{2}\Omega_{2}}\right)\exp\left(-\sqrt{\frac{4\sigma^2cz}{P_{2}\Omega_{1}\Omega_{2}}}\right)}{1{\setlength\arraycolsep{0.5pt}+}z}dz\leq C_{e}\leq C_{e}^{t}\nonumber\\
&\leq\frac{1}{2\ln2}\sum_{i=1}^{2}\Psi\left(1,1;\frac{\sigma^2b}{P_{i}\Omega_{i}}\right).\label{bound_EC}
\end{align}
where we can get $C^{t}_{e}$ by substituting $\mathcal{J}_{i,l}\thickapprox\left(\frac{P_{i}\Omega_{i}}{\sigma^2b}\right)^{l+1}l!\left(\psi(l+1){\setlength\arraycolsep{0.5pt}-}\ln{\frac{\sigma^2b}{P_{i}\Omega_{i}}}\right)$ into (\ref{bound_EC}). Here, $\psi(1)=-\mathbf{C}$ and $\psi(k)=-\mathbf{C}+\sum_{i=1}^{k-1}\frac{1}{i}$ for $k\geqslant2$.
\end{proposition}
\begin{proof}
Firstly, we have
\begin{align}
\mathcal{J}_{i,l}&\leqslant\int_{0}^{\infty}\exp\left(\frac{{\setlength\arraycolsep{0.3pt}-}\sigma^2bz}{P_{i}\Omega_{i}}\right)z^{l}\ln{z}dz\nonumber\\
&=\left(\frac{P_{i}\Omega_{i}}{\sigma^2b}\right)^{l+1}l!\left(\psi(l+1){\setlength\arraycolsep{0.5pt}-}\ln{\frac{\sigma^2b}{P_{i}\Omega_{i}}}\right).
\end{align}
As a result, a tight lower bound of $C^{t}_{e}$ can be obtained 

Applying the upper bound of $xK_{1}(x)$ in (\ref{bound_K1}) and the series expansion of $\mathcal{K}_{1}(x)$, it is easy to calculate the second item of $C_{e}$ is smaller than 0. We thus have $\frac{1}{2\ln2}\sum_{i=1}^{2}\Psi\left(1,1;\frac{\sigma^2b}{P_{i}\Omega_{i}}\right)\geq C_{e}^{t}$. Based on the lower bound of $xK_{1}(x)$ in (\ref{bound_K1}),  we have the lower bound of $C_{e}$ and the proof is completed.
\end{proof}

Clearly, the upper bound closely match with the exact egodic capacity results at both high SNR and $\lambda$ regime because of $x\varpropto \sigma^2/\lambda P_{1}$ and $xK_{1}(x)\rightarrow1$ when $x\rightarrow0$.
\subsection{Finite-SNR DMT}
In this subsection, we characterize the finite-SNR diversity-multiplexing tradeoff for two-way AF relaying channels with energy harvesting. Following \cite{Finite_DMT}, the diversity gain at finite SNR is described by
\setcounter{equation}{26}
\begin{align}
d(r,\gamma)=-\frac{\partial\ln P_{out}}{\partial\ln\gamma}=-\frac{\gamma}{P_{out}}\frac{\partial P_{out}}{\partial\gamma},\label{dr}
\end{align}
where it follows that the multiplexing gain is $r=R/(\frac{1}{2}\log_{2}(1+\gamma))$ for two-way AF relaying channels.

We consider the symmetric relaying scenario, i.e., $P_{1}/\sigma^2=P_{2}/\sigma^2=\gamma$ and $T_{1}=T_{2}=R$, for the finite-SNR DMT. Thus, we have $\tau_{1}=\tau_{2}=(1+\gamma)^{r}-1$. It is difficult to obtain the diversity gain based on the exact expression of $P_{out}$ in (\ref{outage_exact}). Meanwhile, numerical results show that the lower bound in (\ref{lowerbound}) is closed to $P_{out}$ as shown in the next section. By using Proposition 1, the finite-SNR DMT can now be evaluated.
\begin{theorem}
The finite-SNR diversity-multiplexing tradeoff for the proposed protocol with symmetric relaying ($P_{1}/\sigma^2=P_{2}/\sigma^2=\gamma$ and $T_{1}=T_{2}=R$) is given by (\ref{DMT}) at the bottom of this page. Here, we use
\setcounter{equation}{28}
\begin{align}
&A\triangleq\frac{\partial X_{0}}{\partial \gamma}=\Bigg(r\gamma(1+\gamma)^{r-1}-(1+\gamma)^{r}+1\Bigg)\nonumber\\
&\times\left\{\frac{b}{2\gamma^2}\left(1+\sqrt{1+\frac{4c\gamma}{b^2(1+\gamma)^{r}-b^2}}\right)-\frac{c}{\gamma((1+\gamma)^{r}-1)}\right.\nonumber\\
&\left.\times\left(1+\frac{4c\gamma}{b^2(1+\gamma)^{r}-b^2}\right)^{-\frac{1}{2}}\right\},\label{VA}
\end{align}
and
\begin{align}
B\triangleq\frac{\partial \frac{(1+\gamma)^{r}-1}{\gamma}}{\partial \gamma}=\frac{r\gamma(1+\gamma)^{r-1}-(1+\gamma)^{r}+1}{\gamma^2}.\label{VB}
\end{align}
\end{theorem}
\begin{proof}
According to $P_{1}=P_{2}$ and $T_{1}=T_{2}$, $X_{0}$ is equal to $Y_{0}$ and can be simplified as
\begin{align}
X_{0}=\frac{b((1+\gamma)^{r}{\setlength\arraycolsep{0.3pt}-}1)}{2\gamma}\left(1+\sqrt{1+\frac{4c\gamma}{b^2(1+\gamma)^{r}-b^2}}\right).\label{X0}
\end{align}
By substituting (\ref{X0}) into (\ref{lowerbound}), (\ref{VA}) can be calculated. Finally we can obtain (\ref{DMT}).
\end{proof}
\section{Numerical Results}
In this section, numerical results are presented to analyze and verify the accuracy of the derived analytical expressions. The effect of the power splitting ratio $\lambda$ at the energy harvesting relay on the system performance are also discussed. In the simulation, we consider both sources are separated by a normalized distance. Let $d_{1} (1-d_{1})$ denote the distance between $S_{1} (S_{2})$ and the relay node. Considering the large scale path loss, we have $\Omega_{1}=1/d_{1}^{3}$ and $\Omega_{2}=1/(1-d_{1})^{3}$ for $h_{1}$ and $h_{2}$ respectively. We use $d_{1}=1/2$, $T_{1}=T_{2}=1$ bps/Hz, $P_{1}=P_{2}$, $\eta=1$, and $\epsilon=1/2$ for the simulation unless special remark.

In Fig. \ref{outage_simulation} the simulation results are compared with the analytical results in terms of the outage probability. The analytical results for the outage probability are developed through (\ref{outage_exact}) and closely match with the simulation results. The lower and upper bounds derived in Proposition 1 are evaluated in the figure. We see both bounds are tight, especially for the asymmetric relaying traffic, e.g., $d_{1}\neq 1-d_{1}$ and $T_{1}\neq T_{2}$.
 \begin{figure}[t]
\centering
\includegraphics[width=2.8in,height=2.3in]{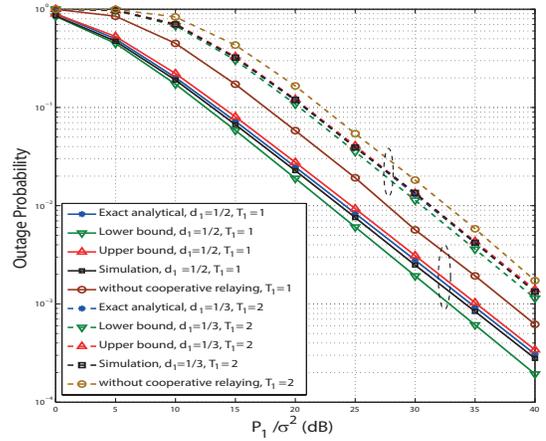}
\caption{Exact and bounds outage probabilities vs. SNR. Here, $\lambda=3/4$.}
\label{outage_simulation}
\end{figure}

 \begin{figure}[t]
\centering
\includegraphics[width=2.8in,height=2.3in]{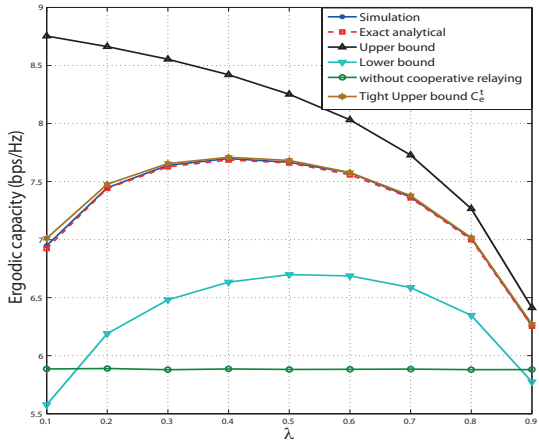}
\caption{Ergodic capacity vs. $\lambda$. Here, $P_{1}/\sigma^2=P_{2}/\sigma^2=20$ dB.}
\label{capacity_simulation}
\end{figure}

Fig. \ref{capacity_simulation} depicts the impact of the power splitting ratio $\lambda$ on the ergodic capacity. According to the definition of $\lambda$, the energy scavenged from the received signal would become more strong with a larger $\lambda$, reducing the forward signal power. We observe from Fig. \ref{capacity_simulation} that for $P_{1}/\sigma^2=P_{2}/\sigma^2=20$ dB, a reasonable value of $\lambda$ is from $0.3$ to $0.6$. Extreme values can significantly degrade the ergodic capacity. It can also be seen that the analytical result is in excellent agreement with the simulation result. The tight upper bound $C_{e}^{t}$ in (\ref{bound_EC}) is very close to the simulation results.  We can also see that the upper and lower bounds of the ergodic capacity is close to the exact ergodic capacity at high $\lambda$ region, while the gap between the bounds and the exact value is large at low $\lambda$ region. This is due to the fact that the exact expression in Theorem 2 is close to the upper bound in Proposition 2 if $xK_{1}(x)\rightarrow1$ when $x\rightarrow0$. In this case, $x\varpropto \sigma^2/\lambda P_{1}$, yielding $x\rightarrow\infty$ when $\lambda\rightarrow0$.

Fig. \ref{DMT_simulation} demonstrates the finite-SNR DMT based on (\ref{DMT}). It is shown that the diversity gain $d$ increases as SNR increases. Fig. \ref{DMT_L_simulation} depicts the diversity gain $d$ as a function of the power splitting ratio $\lambda$ with the multiplexing gain $r=0.5$ for difference $d_{1}$. It is clearly shown that there exists a trade-off between $d$ and $\lambda$. We observe that when the relay node moves towards $S_{1}$ or $S_{2}$, if $\lambda$ is set to small value, e.g., $\lambda=0.1$, the proposed protocol yields the largest diversity gain. When the relay node is placed at the middle position, the largest diversity gain can be obtained if $\lambda$ is around $1/2$.
 \begin{figure}[t]
\centering
\includegraphics[width=2.8in,height=2.3in]{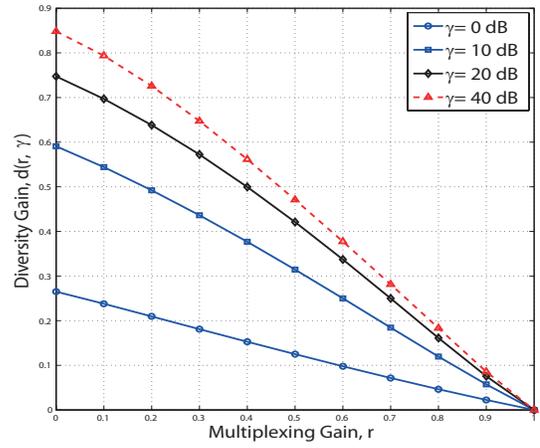}
\caption{Finite SNR DMT of the proposed protocol. Here, $\lambda=3/4$.}
\label{DMT_simulation}
\end{figure}

 \begin{figure}[t]
\centering
\includegraphics[width=2.8in,height=2.3in]{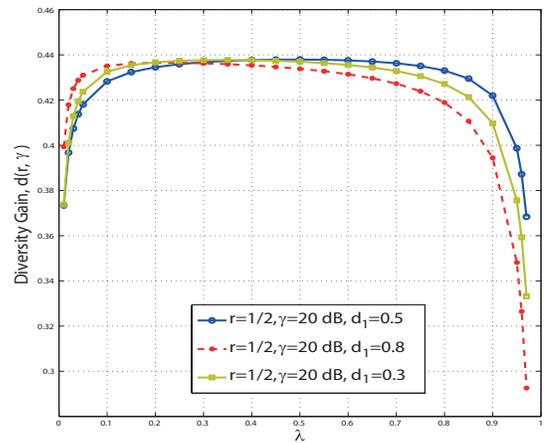}
\caption{Diversity gain vs. $\lambda$. Here, $r=0.5$ and $P_{1}/\sigma^2=P_{2}/\sigma^2=20$ dB.}
\label{DMT_L_simulation}
\end{figure}

Finally, comparing with the non-cooperative relaying scheme in Fig. \ref{outage_simulation} and Fig. \ref{capacity_simulation}, we find that the use of the relay node can improve the ergodic capacity and achieve higher outage performance. Surprisingly, this improvement is not build on the additional resource consumption, neither the energy, time nor radio spectrum resource.
\section{Conclusion}
This paper has developed and analyzed the wireless information and power transfer protocol in two-way AF relaying channels. The proposed protocol not only does not consume additional resource, but also can improve the transmission efficiency and offer more reliable communication. We have characterized the exact expressions of the proposed energy harvesting relaying protocol in terms of outage probability, ergodic capacity and finite-SNR DMT. Besides, we have derived the tight bounds of the outage probability and the ergodic capacity. Numerical results are presented to verify the accuracy of our theoretical predictions. Furthermore, we obtain the valuable insights on the impact of the power splitting ratio $\lambda$ on the ergodic capacity and finite-SNR DMT.

\appendices
\renewcommand{\theequation}{\thesection.\arabic{equation}}
 \setcounter{equation}{0}
\section{Proof of Theorem 1}
 In order to proceed, solving the following equations
\begin{equation}
Y=\frac{\sigma^2\tau_{1}}{P_{2}}(b+\frac{c}{X}),~~X=\frac{\sigma^2\tau_{2}}{P_{1}}(b+\frac{c}{Y}),
\end{equation}
we have the solution point $(X_{0},Y_{0})$. It is easy to see that $Y_{0}>(b+c/X)\sigma^2\tau_{1}/P_{2}>c/(P_{1}X/\sigma^2\tau_{2}-b)$ when $X>X_{0}$. Likewise, we have $X_{0}>(b+c/Y)\sigma^2\tau_{2}/P_{1}>c/(P_{2}Y/\sigma^2\tau_{1}-b)$ when $Y>Y_{0}$. We now can compute
\begin{align}
&\Pr\left(\gamma_{1}<\tau_{1}, \gamma_{2}<\tau_{2}\right)\label{A_1}\\
&\Pr\left(|h_{2}|^2<\frac{\sigma^2\tau_{1}}{P_{2}}(b+\frac{c}{|h_{1}|^2}),|h_{1}|^2<\frac{\sigma^2\tau_{2}}{P_{1}}(b+\frac{c}{|h_{2}|^2})\right)\nonumber\\
&=\int_{0}^{Y_{0}}\int_{\frac{X_{0}}{Y_{0}}y}^{\frac{\sigma^2\tau_{2}}{P_{1}}(b+\frac{c}{y})}\frac{1}{\Omega_{1}}\exp(-\frac{x}{\Omega_{1}})\frac{1}{\Omega_{2}}\exp(-\frac{y}{\Omega_{2}})dxdy\nonumber\\
&{\setlength\arraycolsep{0.5pt}+}\int_{0}^{X_{0}}\int_{\frac{Y_{0}}{X_{0}}x}^{\frac{\sigma^2\tau_{1}}{P_{2}}(b{\setlength\arraycolsep{0.5pt}+}\frac{c}{x})}\frac{1}{\Omega_{1}}\exp(-\frac{x}{\Omega_{1}})\frac{1}{\Omega_{2}}\exp(-\frac{y}{\Omega_{2}})dydx\nonumber\\
&=1-\exp{\left(-\frac{X_{0}}{\Omega_{1}}-\frac{Y_{0}}{\Omega_{2}}\right)}-\sum\limits_{ i,j \in \{ 1,2\}  \hfill \atop
i \ne j }\frac{1}{\Omega_{j}}\exp\left(-\frac{\sigma^2\tau_{j}b}{P_{i}\Omega_{i}}\right)\mathcal{I}_{i},\nonumber
\end{align}
where $\mathcal{I}_{i}=\int_{0}^{V_{i}}\exp\left(\frac{-\sigma^2\tau_{j}c}{P_{i}\Omega_{i}}\frac{1}{z}-\frac{z}{\Omega_{j}}\right)dz$ for $i=1,2$ with $V_{1}=Y_{0}$ and $V_{2}=X_{0}$.

It is difficult to obtain the closed forms of $\mathcal{I}_{1}$ and $\mathcal{I}_{2}$. The second-order Taylor series expansion of $\mathcal{H}_{i,j}(z)$ can be used to obtain an approximation of $\mathcal{I}_{1}$ and $\mathcal{I}_{2}$. Thus we have
\begin{align}
\mathcal{I}_{1}&\approx\sum_{n=0}^{2}\frac{\mathcal{H}^{(n)}_{1,2}(\nu_{1})}{n!}\int_{0}^{Y_{0}}(y-\nu_{1})^{n}dy\nonumber\\
&=\sum_{n=0}^{2}\frac{\mathcal{H}^{(n)}_{1,2}(\nu_{1})}{n!(n+1)}\Bigg\{(Y_{0}-\nu_{1})^{n+1}-(-\nu_{1})^{n+1}\Bigg\}
\end{align}
where $\nu_{1}=Y_{0}/2$ is the convergent point. Likewise, we can also have $\mathcal{I}_{2}$.

Based on Lemma 1, we thus have
\begin{align}
\Pr(\gamma_{i}{\setlength\arraycolsep{0.5pt}<}\tau_{i}){\setlength\arraycolsep{0.5pt}=}1{\setlength\arraycolsep{0.5pt}-}\exp({\setlength\arraycolsep{0.5pt}-}\frac{\sigma^2\tau_{i}b}{P_{j}\Omega_{j}})\sqrt{\frac{4\sigma^2\tau_{i}c}{P_{j}\Omega_{1}\Omega_{2}}}\mathcal{K}_{1}\left(\sqrt{\frac{4\sigma^2\tau_{i}c}{P_{j}\Omega_{1}\Omega_{2}}}\right),\nonumber
\end{align}
for $i,j\in\{1,2\}$ and $i\neq j$.  The theorem is thus proved.
\section{Proof of Proposition 1}
 \setcounter{equation}{0}
Following the proof of Theorem 1, it is easy to calculate
\begin{align}
\mathcal{I}_{1}&=\sqrt{\frac{4\sigma^2\tau_{i}c}{P_{j}\Omega_{1}\Omega_{2}}}\mathcal{K}_{1}\left(\sqrt{\frac{4\sigma^2\tau_{i}c}{P_{j}\Omega_{1}\Omega_{2}}}\right)\nonumber\\
&-\int_{Y_{0}}^{\infty}\exp\left(\frac{-\sigma^2\tau_{2}c}{P_{1}\Omega_{1}}\frac{1}{y}-\frac{y}{\Omega_{2}}\right)dy.\label{B_1}
\end{align}
Then, it can be bounded as following
\begin{align}
&\exp\left(-\frac{Y_{0}}{\Omega_{2}}\right)\int_{0}^{\infty}\exp\left(\frac{{\setlength\arraycolsep{0.5pt}-}\sigma^2\tau_{2}c}{P_{1}\Omega_{1}y}{\setlength\arraycolsep{0.5pt}-}\frac{y}{\Omega_{2}}\right)dy\label{B_2}\\
&\leq\int_{Y_{0}}^{\infty}\exp\left(\frac{{\setlength\arraycolsep{0.5pt}-}\sigma^2\tau_{2}c}{P_{1}\Omega_{1}y}{\setlength\arraycolsep{0.5pt}-}\frac{y}{\Omega_{2}}\right)dy\leq Y_{0}\int_{1}^{\infty}\exp\left(-\frac{zY_{0}}{\Omega_{2}}\right)dz.\nonumber
\end{align}
Similarly, we get the lower and upper bounds of $\int_{X_{0}}^{\infty}\exp\left(\frac{-\sigma^2\tau_{1}c}{P_{2}\Omega_{2}}\frac{1}{x}-\frac{x}{\Omega_{1}}\right)dx$. By substituting (\ref{B_2}) and (\ref{B_1}) into (\ref{A_1}), the proof is finished.
\section{Proof of Theorem 2}
 \setcounter{equation}{0}
Applying the series expansion of $\mathcal{K}_{1}(x)$, we have
 \begin{align}
&\int_{0}^{\infty}\frac{1{\setlength\arraycolsep{0.5pt}-}F_{1}(z)}{1{\setlength\arraycolsep{0.5pt}+}z}dz{\setlength\arraycolsep{0.5pt}=}\int_{0}^{\infty}\frac{\exp\left(\frac{-\sigma^2bz}{P_{2}\Omega_{2}}\right)}{1+z}dz{\setlength\arraycolsep{0.5pt}+}\sum_{l=1}^{\infty}\left(\ln\frac{\sigma^2cz}{P_{2}\Omega_{1}\Omega_{2}}\right.\nonumber\\
&\left.{\setlength\arraycolsep{0.5pt}+}2\mathbf{C}{\setlength\arraycolsep{0.5pt}-}\sum_{k=1}^{l}{\frac{1}{k}}{\setlength\arraycolsep{0.5pt}-}\sum_{k=1}^{l{\setlength\arraycolsep{0.5pt}+}1}{\frac{1}{k}}\right)\int_{0}^{\infty}\frac{\exp\left(\frac{-\sigma^2bz}{P_{2}\Omega_{2}}\right)}{1+z}\frac{\left(\frac{\sigma^2cz}{P_{2}\Omega_{1}\Omega_{2}}\right)^{l{\setlength\arraycolsep{0.5pt}+}1}}{l!(l+1)!}dz.\nonumber\\
&{\setlength\arraycolsep{0.5pt}+}\int_{0}^{\infty}\frac{\exp\left(\frac{-\sigma^2bz}{P_{2}\Omega_{2}}\right)\frac{2\sigma^2cz}{P_{2}\Omega_{1}\Omega_{2}}\left(\ln\sqrt{\frac{\sigma^2cz}{P_{2}\Omega_{1}\Omega_{2}}}{\setlength\arraycolsep{0.5pt}+}\mathbf{C}{\setlength\arraycolsep{0.5pt}-}\frac{1}{2}\right)}{1+z}.
 \end{align}
 
 Let us denote the first, second and third items in the RHS of the above equation as $Q_{1}$, $Q_{2}$ and $Q_{3}$ respectively. Following the definition of the confluent hypergeometric function \cite{Math_book}, $Q_{1}$ is given by
 \begin{align}
Q_{1}=\Psi(1,1;\frac{\sigma^2b}{P_{2}\Omega_{2}})\Gamma(1)=\Psi(1,1;\frac{\sigma^2b}{P_{2}\Omega_{2}}),\label{q1}
\end{align}
where $\Gamma(\cdot)$ is the Gamma function and $\Gamma(1)=1$. $Q_{2}$ can be rewritten as
 \begin{align}
&Q_{2}=\sum_{l=1}^{\infty}\frac{\left(\frac{\sigma^2c}{P_{2}\Omega_{1}\Omega_{2}}\right)^{l{\setlength\arraycolsep{0.5pt}+}1}}{l!(l+1)!}\int_{0}^{\infty}\frac{\exp\left(\frac{-\sigma^2bz}{P_{2}\Omega_{2}}\right)z^{l{\setlength\arraycolsep{0.5pt}+}1}}{1+z}\label{q2}\\
&\times\left\{\left(\ln\frac{\sigma^2c}{P_{2}\Omega_{1}\Omega_{2}}{\setlength\arraycolsep{0.5pt}+}2\mathbf{C}{\setlength\arraycolsep{0.5pt}-}\sum_{k=1}^{l}{\frac{1}{k}}{\setlength\arraycolsep{0.5pt}-}\sum_{k=1}^{l{\setlength\arraycolsep{0.5pt}+}1}{\frac{1}{k}}\right)+\ln{z}\right\}dz\nonumber\\
&{{\setlength\arraycolsep{0.5pt}=}}\sum_{l=1}^{+\infty}\frac{\left(\frac{\sigma^2c}{P_{2}\Omega_{1}\Omega_{2}}\right)^{l{\setlength\arraycolsep{0.5pt}+}1}}{l!}\left\{\left(\ln\frac{\sigma^2c}{P_{2}\Omega_{1}\Omega_{2}}{\setlength\arraycolsep{0.5pt}+}2\mathbf{C}{\setlength\arraycolsep{0.5pt}-}\sum_{k=1}^{l}{\frac{1}{k}}{\setlength\arraycolsep{0.5pt}-}\sum_{k=1}^{l{\setlength\arraycolsep{0.5pt}+}1}{\frac{1}{k}}\right)\right.\nonumber\\
&\left.{\setlength\arraycolsep{0.3pt\times}}\Psi\left(l{\setlength\arraycolsep{0.3pt+}}2,l{\setlength\arraycolsep{0.3pt+}}2;\frac{\sigma^2b}{P_{2}\Omega_{2}}\right){\setlength\arraycolsep{0.3pt}+}\int_{0}^{\infty}\frac{\exp\left(\frac{{\setlength\arraycolsep{0.3pt}-}\sigma^2bz}{P_{2}\Omega_{2}}\right)z^{l{\setlength\arraycolsep{0.3pt}{\setlength\arraycolsep{0.3pt}+}}1}\ln{z}}{(l+1)!(1+z)}dz\right\}.\nonumber
 \end{align}
where we use $\Gamma(n+1)=n!$ for a natural number $n$ in the last step. Likewise, we have 
\begin{align}
&Q_{3}=\frac{\sigma^2c}{P_{2}\Omega_{1}\Omega_{2}}\left(\ln\frac{\sigma^2c}{P_{2}\Omega_{1}\Omega_{2}}{\setlength\arraycolsep{0.5pt}+}2\mathbf{C}{\setlength\arraycolsep{0.5pt}-}1\right)\Psi\left(2,2;\frac{\sigma^2b}{P_{2}\Omega_{2}}\right)\nonumber\\
&+\frac{\sigma^2c}{P_{2}\Omega_{1}\Omega_{2}}\int_{0}^{\infty}\frac{\exp\left(\frac{{\setlength\arraycolsep{0.3pt}-}\sigma^2bz}{P_{2}\Omega_{2}}\right)z\ln{z}}{1+z}dz.
\end{align}

Consequently, the closed-form expression of $\int_{0}^{\infty}\frac{1{\setlength\arraycolsep{0.5pt}-}F_{2}(z)}{1{\setlength\arraycolsep{0.5pt}+}z}dz$ can also be derived in the same way.
\bibliographystyle{IEEEtran}
\bibliography{IEEEabrv,overall}
\end{document}